\documentclass[final]{siamart1116}

\usepackage{lipsum}
\usepackage{amssymb,amsfonts}
\usepackage{mathtools}
\usepackage{float}
\usepackage{tikz}
\usepackage{bm}
\usepackage[caption=false]{subfig}
\usepackage{dsfont}
\usepackage{graphicx}
\usepackage{enumerate}
\usepackage{algpseudocode}
\usepackage{soul}
\usepackage{algorithmicx}
\usepackage{anyfontsize}
\ifpdf
  \DeclareGraphicsExtensions{.eps,.pdf,.png,.jpg}
\else
  \DeclareGraphicsExtensions{.eps}
\fi

\algnewcommand{\LeftComment}[1]{\Statex \(\triangleright\) #1}

\setstcolor{blue}

\numberwithin{theorem}{section}

\newcommand{\TheTitle}{Symmetry-independent stability analysis of synchronization patterns}
\newcommand{\TheAuthors}{Y. Zhang and A. E. Motter}
\newcommand{\ShortTitle}{Symmetry-independent analysis of synchronization patterns}

\headers{\ShortTitle}{\TheAuthors}

\title{{\TheTitle}\thanks{Published version at \textit{SIAM Rev.} \textbf{62} pp. 817–836 (\url{https://doi.org/10.1137/19M127358X}).
\funding{This work was funded by ARO Grant No.~W911NF-19-1-0383.}}}

\author{
  Yuanzhao Zhang\thanks{
    Department of Physics and Astronomy, Northwestern University, Evanston, Illinois 60208, USA
    (\email{yuanzhao@u.northwestern.edu}).}
  \and
  Adilson E. Motter\thanks{
    Department of Physics and Astronomy and Northwestern Institute on Complex Systems, Northwestern University, Evanston, Illinois 60208, USA 
    (\email{motter@northwestern.edu}).}
}

\usepackage{amsopn}

\ifpdf
\hypersetup{
  pdftitle={\TheTitle},
  pdfauthor={\TheAuthors}
}
\fi

\begin{document}

\maketitle

\begin{abstract}
  The field of network synchronization has seen tremendous growth following the introduction of the master stability function (MSF) formalism, which enables the efficient stability analysis of synchronization in large oscillator networks. 
  However, to make further progress we must overcome the limitations of this celebrated formalism, which focuses on global synchronization and requires both the oscillators and their interaction functions to be identical, while many systems of interest are inherently heterogeneous and exhibit complex synchronization patterns. 
  Here, we establish a generalization of the MSF formalism that can characterize the stability of any cluster synchronization pattern, even when the oscillators and/or their interaction functions are nonidentical.
  The new framework is based on finding the finest simultaneous block diagonalization of matrices in the variational equation and does not rely on information about network symmetry.
  This leads to an algorithm that is error-tolerant and orders of magnitude faster than existing symmetry-based algorithms. 
  As an application, we rigorously characterize the stability of chimera states in networks with multiple types of interactions.
\end{abstract}



\section{Introduction}
\label{sec:intro}

Coupled oscillator networks have been extensively studied as a fundamental model of collective behavior in complex systems \cite{strogatz2000kuramoto,boccaletti2002synchronization,newman2003structure,boccaletti2006complex,arenas2008synchronization,cao2013overview,abrams2016introduction}.
The field is unique in its close interaction between theoretical developments \cite{pecora1990synchronization,rosenblum1996phase,timme2002prevalence,ott2008low,aguiar2011dynamics} and practical applications \cite{wiesenfeld1996synchronization,feki2003adaptive,belykh2005synchronization,li2010consensus,motter2013spontaneous}.
A central theme of current research is how to characterize the stability of increasingly complex synchronization patterns in arbitrary network structures.
Such patterns can be regarded as forms of cluster synchronization, in which the oscillators form two or more internally synchronized clusters that exhibit mutually distinct dynamics \cite{belykh2001cluster,golubitsky2005patterns,belykh2008cluster,dahms2012cluster,rosin2013control,williams2013experimental,orosz2014decomposition,hart2017experiments}.
The stability of a synchronization pattern is important because it usually cannot be observed if it is unstable, and thus a bifurcation leading to the loss or restoration of stability has significant ramifications in various biological and technological systems.

In order to perform efficient stability analysis for large networks of coupled oscillators, the key is to first divide the full state space of the variational equation into minimal flow-invariant subspaces (defined below) and then calculate the maximum Lyapunov exponent in each flow-invariant subspace to determine whether perturbations within that subspace will grow.
To achieve this for the global synchronization of identical oscillators, the master stability function (MSF) formalism \cite{pecora1998master} finds a coordinate transformation that diagonalizes the coupling matrix, which in turn decouples the high-dimensional variational equation of the full network into low-dimensional equations describing the evolution of independent perturbation modes. 
The full equation has a dimension that grows linearly with the network size, while the decoupled equations all have a fixed dimension equal to that of an individual oscillator, irrespective of the network size.

However, when one considers cluster synchronization states, nonidentical oscillators, or disparate interactions, all of which are common in real systems, there are in general two or more noncommuting matrices in the variational equation.
Since noncommuting matrices cannot be diagonalized simultaneously (even when the individual matrices can), the MSF formalism is not applicable to these cases.
The goal of the current paper is to introduce an extension of the MSF formalism and propose a fundamentally new framework based on the theory of matrix $*$-algebra that addresses these important cases.
In particular, we present a highly scalable algorithm that finds the finest simultaneous block diagonalization (SBD) of any given set of self-adjoint matrices, leading to an optimal separation of the perturbation modes and efficient stability analysis of arbitrary synchronization patterns.
The finest SBD has a common block structure characterized by the largest number of blocks and is unique up to block permutations.

Our framework applies to the general class of network dynamical systems described by
\begin{equation}
  \dot{\bm{x}}_i = \bm{F}_i(\bm{x}_i) + \sum_{r=1}^{R} \sigma_r \sum_{j=1}^{n} \bm{C}_{r}(i,j) \bm{H}_r(\bm{x}_i,\bm{x}_j),\quad i=1,\dots,n,
  \label{eq:general}
\end{equation}
where $\bm{x}_i$ is the $d$-dimensional state vector of the $i$th oscillator, $n$ is the number of oscillators, $R$ is the number of interaction types, and the overdot represents the time derivative. 
Here, $\bm{F}_i: \mathbb{R}^d \rightarrow \mathbb{R}^d$ 
is the vector field governing the uncoupled dynamics of the $i$th oscillator and $\bm{C}_r$ is a self-adjoint coupling matrix representing interactions of the form $\bm{H}_r$ and strength $\sigma_r$.
The synchronization patterns we study can be derived from any balanced equivalence relation \cite{golubitsky2005patterns,kamei2013computation}, which is the most general class of patterns for which oscillators in the same cluster can admit equal dynamics for generic $\bm{F}_i$ and $\bm{H}_r$. 
In general, nodes in a cluster can be separated by nodes from other clusters and do not necessarily form a connected component of the network. 
The resulting synchronization patterns describe a wide range of network dynamics, including remote synchronization and chimera states.\footnote{For ease of presentation, we ground our discussions in this paper on \cref{eq:general}, but we note that it is straightforward to generalize our methods beyond ODE settings.
For instance, it is possible to introduce coupling delay into the interaction functions and the resulting delay differential equations can still be analyzed within our framework.
Naturally, the framework also applies to discrete-time dynamical systems.}

A related extension of the MSF formalism to study cluster synchronization patterns was previously proposed by Pecora and colleagues \cite{pecora2014cluster}.
That framework was originally developed for networks with adjacency-matrix coupling and has since been extended to diffusively coupled networks \cite{sorrentino2016complete} and multilayer networks \cite{blaha2019cluster}.
In those studies, the authors simplified the stability analysis using the machinery of irreducible representations (IRRs) \cite{tinkham2003group}, which decouples the variational equation according to symmetries present in the system.

Both the IRR framework and our SBD framework reduce to the MSF formalism for global synchronization of identical oscillators with a single type of interaction.
The key difference between the two frameworks is that the former relies on network symmetry to perform stability analysis, whereas the latter does not.\footnote{We note that it is computationally inexpensive to identify synchronization patterns when compared to the cost of determining their stabilities. We thus assume that the synchronization patterns of interest are given before stability analyses are performed.}
As a result, the IRR framework has to resort to ad hoc modifications when a cluster synchronization pattern is not induced by network symmetry \cite{sorrentino2016complete,siddique2018symmetry}.
In contrast, our SBD framework does not require any symmetry information to be known in advance and is directly applicable to all cluster synchronization patterns.
Moreover, it forgoes the calculation on IRRs of network symmetry, which becomes computationally prohibitive very quickly as the number of symmetries grows.
This leads to a faster, simpler, and more robust algorithm than existing ones based on the IRR framework and enables the study of complex synchronization patterns in large networks.

The paper is organized as follows. 
In \cref{sec:star-algebra}, we present the concept of matrix $*$-algebra and a fast algorithm for finding the finest SBD for any set of self-adjoint matrices. 
Then, in \cref{sec:cluster}, we develop a symmetry-independent framework for the stability analysis of arbitrary cluster synchronization patterns and compare our algorithm with state-of-the-art algorithms based on IRRs. 
We further show in \cref{sec:unified} that our algorithm can be applied to analyze cluster synchronization patterns of nonidentical oscillators and oscillators with multiple types of interactions. 
The strength of this unified framework is demonstrated with the characterization of permanently stable chimera-like states in multilayer networks.
A discussion on open problems and future directions is presented in \cref{sec:directed}.

\section{Finest simultaneous block diagonalization}
\label{sec:star-algebra}

Given a set of $n \times n$ matrices $\mathcal{B} = \{\bm{B}_1,\cdots,\bm{B}_K\}$, we say that a subspace $\mathcal{W}$ of $\mathbb{C}^n$ is {\it invariant} under $\mathcal{B}$ if $\bm{B}_k\mathcal{W} \subseteq \mathcal{W}$ for every $\bm{B}_k \in \mathcal{B}$. 
Further, an invariant subspace $\mathcal{W}$ is {\it minimal} if no proper subspace of $\mathcal{W}$ other than $\bm{0}$ is invariant under $\mathcal{B}$.
An invertible matrix $\bm{T}$ is said to give the {\it finest SBD} of the matrix set $\mathcal{B}$ if it brings all matrices in $\mathcal{B}$ into a common block-diagonal form that cannot be further refined.
Equivalently, $\bm{T}$ decomposes $\mathbb{C}^n$ into minimal invariant subspaces under $\mathcal{B}$, such that the $j$th common blocks in $\bm{T}^{-1}\mathcal{B}\bm{T}$ only have $\bm{0}$ and $\mathbb{C}^{n_j}$ as invariant subspaces, where $n_j$ is the dimension of the $j$th blocks.

To make progress in finding the finest SBD, it is beneficial to consider an algebraic structure called matrix $*$-algebra. 
Letting $\mathcal{M}_n$ denote the set of all $n \times n$ matrices with complex entries, a subset $\mathcal{T}$ of $\mathcal{M}_n$ is said to be a {\it matrix $*$-algebra} over $\mathbb{C}$ if the identity matrix $\bm{I}_n$ belongs to $\mathcal{T}$ and
\begin{equation}
  \bm{B},\,\bm{C} \in \mathcal{T};\,\alpha,\,\beta \in \mathbb{C} \implies \alpha\bm{B}+\beta\bm{C},\,\bm{BC},\,\bm{B}^*
  \in \mathcal{T},
  \label{eq:def}
\end{equation}
where $*$ denotes conjugate transpose.\footnote{The results remain applicable if the matrix $*$-algebras are over $\mathbb{R}$, as in various examples considered throughout the paper. However, working in $\mathbb{C}$ both allows complex coupling matrices, which arise for oscillator networks that are naturally expressed using complex vector fields (such as coupled Stuart--Landau oscillators), and can lead to finer block structures when networks are directed.} 
Matrix $*$-algebras enjoy better properties than matrix algebras because they are closed under the conjugate transpose operation. 
This makes matrix $*$-algebras semisimple and thus characterizable by the Artin--Wedderburn theorem \cite{lam2013first}.

According to structure theorems based on the Artin--Wedderburn theorem (Theorems~3.1 and 6.1 in \cite{murota2010numerical}), a matrix $*$-algebra $\mathcal{T}$ can always be decomposed through a unitary transformation $\bm{P}$ into a direct sum of $\ell$ irreducible matrix $*$-algebras $\mathcal{T}_j$, 
\begin{equation}
  \bm{P}^*\mathcal{T}\bm{P} = \bigoplus_{j=1}^{\ell} \left( \bm{I}_{m_j} \otimes \mathcal{T}_j \right) = \text{diag}\{ \bm{I}_{m_1} \otimes \mathcal{T}_1, \cdots, \bm{I}_{m_\ell} \otimes \mathcal{T}_\ell \},
  \label{eq:A-W}
\end{equation}
where $\mathcal{T}_j \subseteq \mathcal{M}_{n_j}$, $m_j$ is the multiplicity of $\mathcal{T}_j$, and $\sum_{j=1}^\ell n_jm_j = n$. 
The $\otimes$ symbol denotes the tensor product of matrices (i.e., the Kronecker product), and the summand $\bm{I}_{m_j} \otimes \mathcal{T}_j = \{ \bigoplus_{k=1}^{m_j} \bm{B}: \bm{B} \in \mathcal{T}_j \}$ represents $m_j$ copies of the irreducible matrix $*$-algebra $\mathcal{T}_j$ arranged diagonally (not to be confused with $\bigoplus_{k=1}^{m_j} \mathcal{T}_j$).
We say that a matrix $*$-algebra $\mathcal{T}_j$ is {\it irreducible} if it contains matrices that only share trivial invariant subspaces (i.e., $\bm{0}$ and $\mathbb{C}^{n_j}$). 
\Cref{eq:A-W} is the canonical form of an irreducible decomposition of a matrix $*$-algebra, which is unique up to block permutations and unitary transformations within each block. 
As a consequence, all the matrices in $\mathcal{T}$ can be transformed simultaneously into a block-diagonal form of $\sum_{j=1}^\ell m_j$ blocks through a single unitary matrix $\bm{P}$.

A matrix $*$-algebra $\mathcal{T}$ is said to be {\it generated by} a set of matrices $\mathcal{B}$ if $\mathcal{B} \subseteq \mathcal{T}$ and every matrix in $\mathcal{T}$ can be constructed from $\bm{I}_n$ and $\mathcal{B}$ using the operations of matrix $*$-algebras (i.e., scalar multiplication, matrix addition, matrix multiplication, and conjugate transpose).
In order to calculate a transformation matrix $\bm{P}$ that gives the finest SBD of all matrices in $\mathcal{T}$, we propose a new algorithm (\cref{alg:1}) and refer to the corresponding coordinate transformation as an SBD transformation.
The algorithm involves only numerical linear-algebraic calculations and does not require any algebraic structure (e.g., symmetries) to be known in advance.

\begin{algorithm}
  \caption{
  Finding an SBD transformation for a matrix $*$-algebra generated by a set of $n\times n$ matrices $\mathcal{B} = \{\bm{B}_k\}$ in $\mathcal{O}(n^3)$. (A MATLAB implementation is available as a GitHub repository at \url{https://github.com/y-z-zhang/net-sync-sym/}.)
  }

  \begin{algorithmic}[1]

    \LeftComment{step 1: generate a self-adjoint matrix $\bm{B}$ by combining matrices in $\mathcal{B}$ with random coefficients $c_k$ and $d_k$}

    \State $\bm{B} = \sum_{k=1}^K [c_k (\bm{B}_k + \bm{B}_k^*) + \mathrm{i} d_k (\bm{B}_k - \bm{B}_k^*)]$

    \State find the eigenvectors $\{\bm{v}_j\}$ of $\bm{B}$

    \LeftComment{steps 3 to 9: find (a basis of) the minimal invariant subspace that contains $\bm{v}_1$}

    \State perform Gram--Schmidt orthonormalization on $\{\bm{v}_1, \bm{B}_k\bm{v}_1, \bm{B}_k^*\bm{v}_1\}$, $k = 1,\dots,K$, to obtain a set of orthonormal vectors $\mathcal{V}$

    \State let $\bm{v}$ be a random linear combination of the vectors from $\mathcal{V}$

    \While{the images $\{\bm{B}_k\bm{v}, \bm{B}_k^*\bm{v}\}$, $k = 1,\dots,K$ include vectors that are linearly independent from $\mathcal{V}$}
      \State make these new vectors orthonormal to $\mathcal{V}$ and to each other 
      \State expand $\mathcal{V}$ to include the new vectors
      \State let $\bm{v}$ be a random combination of the vectors from the expanded $\mathcal{V}$
    \EndWhile

    \State let $\bm{P}$ be a matrix whose columns are made of vectors from $\mathcal{V}$

    \LeftComment{steps 11 to 16: find the rest of the minimal invariant subspaces}

    \While{the matrix $\bm{P}$ has less than $n$ columns}
      \State find an eigenvector $\bm{v}_j$ outside the span of $\bm{P}$'s column vectors
      \State make $\bm{v}_j$ orthonormal to the column vectors of $\bm{P}$
      \State repeat steps 3 to 9 with $\bm{v}_1$ replaced by $\bm{v}_j$
      \State add the vectors from $\mathcal{V}$ to $\bm{P}$ as additional columns
    \EndWhile

  \end{algorithmic}
  \label{alg:1}
\end{algorithm}

The idea behind the algorithm is simple. First, pick an eigenvector $\bm{v}_1$ of a self-adjoint matrix $\bm{B} = \sum_{k=1}^{K} [c_k (\bm{B}_k + \bm{B}_k^*) + \mathrm{i} d_k (\bm{B}_k - \bm{B}_k^*)]$, where $c_k$ and $d_k$ are random coefficients drawn from a Gaussian distribution. 
This eigenvector lies inside one of the minimal invariant subspaces of $\mathcal{B}$ with probability $1$.
Furthermore, all the images of $\bm{v}_1$ under $\{\bm{B}_k\}$ and $\{\bm{B}_k^*\}$ must also be inside the same minimal invariant subspace.
By running the Gram--Schmidt process on $\{\bm{v}_1,\bm{B}_1 \bm{v}_1, \bm{B}_1^* \bm{v}_1, \cdots, \bm{B}_K \bm{v}_1, \bm{B}_K^* \bm{v}_1\}$ and discarding the linearly redundant vectors, we can obtain a set of orthonormal vectors all inside the same minimal invariant subspace.
If these vectors span the entire invariant subspace, then we have discovered a common block and can repeat the process starting from another eigenvector of $\bm{B}$ outside the discovered minimal invariant subspace.
Otherwise, we generate a random linear combination $\bm{v}$ of the existing orthonormal vectors and ``explore'' the invariant subspace further by generating images of $\bm{v}$ under $\{\bm{B}_k\}$ and $\{\bm{B}_k^*\}$.
It is easy to see that a complete basis for a minimal invariant subspace can always be reached after no more than $n$ such iterations.
The computational complexity of the algorithm scales as $\mathcal{O}(n^3)$---it can easily handle $n \times n$ matrices with $n$ in the range of thousands.\footnote{It takes $\mathcal{O}(n^2)$ operations to calculate the image of $\bm{v}$ under $\bm{B}_k$, and $\mathcal{O}(n)$ such images need to be computed to discover the transformation matrix $\bm{P}$.}
This distinguishes \cref{alg:1} from the best competing algorithms available \cite{murota2010numerical,maehara2010numerical,maehara2011algorithm}, which run in $\mathcal{O}(n^4)$ time and are already slow at $n \approx 100$.
Moreover, applications of those algorithms to network synchronization have been limited to the study of global synchronization \cite{irving2012synchronization,zhang2017identical}, while our framework enables application of the new algorithm to general synchronization patterns.

In most cases, we are interested in a given set of matrices instead of the full matrix $*$-algebra. 
When is \cref{alg:1} guaranteed to find the finest SBD for a given matrix set $\mathcal{B}$?
A sufficient condition is that the matrices in $\mathcal{B}$ are self-adjoint.

\begin{proposition} 
Given a set of $n \times n$ self-adjoint matrices $\mathcal{B} = \{\bm{B}_1,\cdots,\bm{B}_K\}$, let $\mathcal{T}$ be the matrix $*$-algebra generated by $\mathcal{B}$. 
If a unitary matrix $\bm{P}$ leads to an irreducible decomposition of $\mathcal{T}$, then it also gives rise to the finest SBD of $\mathcal{B}$. 
\label{thm:1}
\end{proposition}
\begin{proof}
Assume that an invertible matrix $\bm{T}$ yields the finest SBD of the set of self-adjoint matrices $\mathcal{B}$.
Since 
\[
  \bm{T}^{-1}(\alpha\bm{B}_j+\beta\bm{B}_k)\bm{T} = \alpha\bm{T}^{-1}\bm{B}_j\bm{T} + \beta\bm{T}^{-1}\bm{B}_k\bm{T},
\]
\[
  \bm{T}^{-1}\bm{B}_j\bm{B}_k\bm{T} = (\bm{T}^{-1}\bm{B}_j\bm{T})(\bm{T}^{-1}\bm{B}_k\bm{T}),
\]
\[
  \bm{T}^{-1}\bm{B}_j^*\bm{T} = \bm{T}^{-1}\bm{B}_j\bm{T} \text{ for self-adjoint } \bm{B}_j,
\]
all matrices in $\bm{T}^{-1}\mathcal{T}\bm{T}$ can admit the same block structure shared by the matrices in its generating set $\bm{T}^{-1}\mathcal{B}\bm{T}$.
Since this block structure is finest in $\mathcal{B}$, it is also finest in $\mathcal{T}$.
By definition, the unitary matrix $\bm{P}$ yields the finest SBD of the matrices in $\mathcal{T}$.
Because $\mathcal{B}$ and $\mathcal{T}$ have the same finest common block structure, $\bm{P}$ also generates the finest SBD of the matrix set $\mathcal{B}$.
\end{proof}

Taken together, \Cref{thm:1,alg:1} establish a powerful framework for finding the finest SBD for any set of self-adjoint matrices. 
In the sections below, we show how SBD transformations can be used to characterize the stability of synchronization patterns in an efficient and unified fashion. 
We will focus mainly on oscillators coupled through undirected networks (i.e., self-adjoint coupling matrices). 
The possibility of extending the current framework to directed networks will be discussed in \cref{sec:directed}.

\section{Cluster synchronization from a symmetry-independent perspective}
\label{sec:cluster}

Consider a network of $n$ identical $d$-dimensional oscillators forming a synchronization pattern of $M$ clusters.
The cluster synchronization subspace can be defined as an $Md$-dimensional subspace of the full $nd$-dimensional state space, in which oscillators from the same cluster have exactly the same dynamics.
{\it Parallel perturbations} are perturbations inside the cluster synchronization subspace---they do not destroy the cluster synchronization pattern.
{\it Transverse perturbations} are those that are perpendicular to the cluster synchronization subspace---all of them must have negative Lyapunov exponents in order for the cluster synchronization pattern to be stable.

One key step in analyzing the stability of a synchronization pattern amounts to finding a coordinate system that separates the evolution of transverse and parallel perturbation modes. 
The coordinate transformation should also decouple the transverse perturbation modes as much as possible.
The current state-of-the-art method exploits symmetries in the network structure and uses the machinery of group representation theory \cite{pecora2014cluster}. 
In this section, we establish a symmetry-independent framework based on SBD transformations and compare \cref{alg:1} with symmetry-based algorithms in terms of speed, simplicity, and versatility. 

\subsection{The symmetry perspective}

A network of $n$ identical oscillators with adjacency-matrix coupling can be described as the following special case of \cref{eq:general}:
\begin{equation}
  \dot{\bm{x}}_i = \bm{F}(\bm{x}_i) + \sigma \sum_{j=1}^{n} A(i,j) \bm{H}(\bm{x}_j),
  \label{eq:adj-cluster}
\end{equation}
where $\bm{A} = \{A(i,j)\}$ is the self-adjoint adjacency matrix encoding the structure of the underlying network.

To study the stability of cluster synchronization states, it is necessary to first identify possible synchronization patterns supported by \cref{eq:adj-cluster}, a subset of which is determined by the symmetries of the network. 
The network symmetries, described by the graph automorphism group $\mathrm{Aut}(\bm{A})$, can be computed using discrete algebra software \cite{stein2008sage}. 
The nodes can be partitioned into disjoint clusters: two nodes belong to the same cluster if there is a symmetry operation (i.e., node permutations that respect the adjacency matrix) from the automorphism group that maps one node to the other.
In other words, nodes are partitioned according to the orbits under the action of $\mathrm{Aut}(\bm{A})$ \cite{pecora2014cluster}. 
This is the coarsest synchronization pattern that can be derived from network symmetry. 
If one instead considers a subgroup $G$ of $\mathrm{Aut}(\bm{A})$, the nodes can then be partitioned into finer clusters according to the orbits under the action of $G$. 
We call these partitions the {\it orbital partitions} of the network and refer to the corresponding clusters as {\it symmetry clusters}.
Each element $g \in G$ can be represented by a permutation matrix $\bm{R}_g$, upon whose action the adjacency matrix $\bm{A}$ stays invariant, i.e., $\bm{R}_g^* \bm{A} \bm{R}_g = \bm{A}$. 
The set of matrices $\{\bm{R}_g\}$ forms a {\it permutation representation} of the subgroup $G$. 
As a result of network symmetry, nodes in each symmetry cluster receive the same input from other clusters and admit equal dynamics. 
In other words, cluster synchronization patterns based on symmetry clusters are guaranteed to be {\it flow invariant} (i.e., subspaces of the state space that are invariant under time evolution of the system).

Once a subgroup $G$ and the corresponding orbital partition have been identified, one can find the associated cluster synchronization state by evolving the dynamical equation on a quotient network in which each symmetry cluster is represented by a single node. 
\Cref{eq:adj-cluster} can then be linearized around the cluster synchronization state, leading to a variational equation that determines the stability of the corresponding synchronization pattern:
\begin{equation}
  \begin{split}
    \delta\dot{\bm{X}} = & \left( \sum_{m=1}^M \bm{E}_m \otimes J\bm{F}(\bm{s}_m) + \sigma \left( \bm{A} \otimes \bm{I}_d \right) \sum_{m=1}^M \bm{E}_m \otimes J\bm{H}(\bm{s}_m) \right) \delta\bm{X}, \\
    = & \left( \sum_{m=1}^M \bm{E}_m \otimes J\bm{F}(\bm{s}_m) + \sigma \sum_{m=1}^M \bm{A}\bm{E}_m \otimes J\bm{H}(\bm{s}_m) \right) \delta\bm{X},
  \label{eq:var-adj-cluster}
  \end{split}
\end{equation}
where $\bm{s}_m$ is the synchronization trajectory of the $m$th cluster, $\delta\bm{X} = (\delta \bm{x}_1^\intercal, \cdots, \delta \bm{x}_n^\intercal)^\intercal$ is the $nd$-dimensional perturbation vector, and $J$ is the Jacobian operator. 
Let $\mathcal{C}_m$ denote the set of nodes in the $m$th cluster. 
Then
\[
  \bm{E}_m(i,i) =
  \begin{cases}
    1 & \quad \text{if } i \in \mathcal{C}_m, \\
    0 & \quad \text{otherwise} \\
  \end{cases}
\]
is an $n \times n$ diagonal matrix encoding the nodes in the $m$th cluster.
It follows that $\sum_{m=1}^M \bm{E}_m = \bm{I}_n$.

A key insight from \cite{pecora2014cluster} is that there exists a coordinate choice under which the transformed adjacency matrix $\tilde{\bm{A}} = \bm{Q}^*\bm{A}\bm{Q}$ has a block-diagonal form that matches the cluster structure. 
The authors termed this choice the IRR coordinates since the transformation matrix $\bm{Q}$ decomposes the permutation representation $\{\bm{R}_g\}$ into the direct sum of IRRs of $G$.
In particular,
\begin{equation}
  \tilde{\bm{R}}_g = \bm{Q}^*\bm{R}_g\bm{Q} = \bigoplus_{j=1}^{\ell} \left(\tilde{\bm{R}}_g^{(j)} \otimes \bm{I}_{n_j} \right), \quad
  \tilde{\bm{A}} = \bm{Q}^*\bm{A}\bm{Q} = \bigoplus_{j=1}^{\ell} \left( \bm{I}_{m_j} \otimes \tilde{\bm{A}}^{(j)} \right),
  \label{eq:irr}
\end{equation}
where $\ell$ is the number of distinct IRRs present in $\{\tilde{\bm{R}}_g\}$, the $j$th block $\tilde{\bm{A}}^{(j)}$ is an $n_j \times n_j$ matrix with $n_j$ equal to the multiplicity of the $j$th IRR $\{\tilde{\bm{R}}_g^{(j)}\}$, and $m_j$ is the multiplicity of $\tilde{\bm{A}}^{(j)}$ as well as the dimension of $\tilde{\bm{R}}_g^{(j)}$. 
The trivial IRR (which maps every $g \in G$ to 1) is always present with multiplicity $M$, so there is always an $M \times M$ block in $\tilde{\bm{A}}$ corresponding to the dynamics inside the cluster synchronization subspace \cite{pecora2014cluster}. 
In this way, $\bm{Q}$ completely decouples the transverse perturbations from the parallel ones and also separates some of the transverse perturbation modes.

\subsection{The symmetry-independent perspective}

The IRR transformation decouples \cref{eq:var-adj-cluster} by exploiting the network symmetry and its IRRs.
The end result of this transformation is a block diagonalization of the matrix set $\mathcal{A} = \{\bm{E}_1,\cdots,\bm{E}_M,\bm{A}\}$.
But is finding IRRs the most effective way to block diagonalize these matrices?
Below we show that the answer is negative.

Readers might have noticed the parallel between the block forms in \cref{eq:A-W} and \cref{eq:irr}, which hints at a deep connection between the IRR and the SBD transformations. 
We now make this parallel precise by looking at the IRR transformation through the lens of matrix $*$-algebras.

To proceed, we introduce the {\it commutant algebra} $\mathcal{T}'$ of a matrix $*$-algebra $\mathcal{T} \subseteq \mathcal{M}_n$, which is the set of all matrices $\bm{C} \in \mathcal{M}_n$ that commute with every element in $\mathcal{T}$. Letting $[\bm{B},\,\bm{C}] = \bm{B}\bm{C} - \bm{C}\bm{B}$, one has
\begin{equation}
  \mathcal{T}' = \{ \bm{C} \in \mathcal{M}_n \ | \ [\bm{B},\bm{C}] = \bm{0} \ \forall \ \bm{B}\in \mathcal{T} \}.
\end{equation}
$\mathcal{T}'$ is also a matrix $*$-algebra and enjoys the following dual relations with $\mathcal{T}$ \cite{maehara2011algorithm}: 
\begin{itemize}
  \item[] (a) $\mathcal{T}''=\mathcal{T}$, known as the double commutant property.
  \item[] (b) If the irreducible decomposition of $\mathcal{T}$ has blocks of sizes $n_j$ and multiplicities $m_j$, then the irreducible decomposition of $\mathcal{T}'$ is the direct sum of blocks of sizes $m_j$ and multiplicities $n_j$.
\end{itemize}

Given a subgroup $G$ of $\mathrm{Aut}(\bm{A})$, let $\mathcal{S}$ be the set of all $n \times n$ matrices over $\mathbb{C}$ that are invariant under the action of $G$. That is, 
\begin{equation}
  \mathcal{S} = \{ \bm{S} \in \mathcal{M}_n \ | \ \bm{R}_g^* \bm{S} \bm{R}_g = \bm{S} \ \forall \ g \in G  \}.
  \label{eq:define_s}
\end{equation}
First, we note that $\mathcal{S}$ is a matrix $*$-algebra. 
Second, we make a key observation involving the commutant algebra $\mathcal{S}'$. 
It is clear from \cref{eq:define_s} that $\mathcal{S}$ is the commutant algebra of the matrix $*$-algebra $\mathcal{R}$ generated by $\{\bm{R}_g\}$. 
According to the double commutant property, we have $\mathcal{S}'= \mathcal{R}'' = \mathcal{R}$, and hence $\{\bm{R}_g\}$ is a generating set of $\mathcal{S}'$.
Since the IRR transformation $\bm{Q}$ decomposes $\{\bm{R}_g\}$ into the form $\tilde{\bm{R}}_g = \bm{Q}^*\bm{R}_g\bm{Q} = \bigoplus_{j=1}^{\ell} \left(\tilde{\bm{R}}_g^{(j)} \otimes \bm{I}_{n_j} \right)$, the irreducible decomposition of $\mathcal{S}'$ has blocks of sizes $m_j$ and multiplicities $n_j$.\footnote{For any pair of square matrices $\bm{X}$ and $\bm{Y}$ there exists a permutation matrix $\bm{O}$ such that $\bm{X}\otimes\bm{Y} = \bm{O}^\intercal(\bm{Y}\otimes\bm{X})\bm{O}$.}

Next, we utilize the relation $\bm{S}\bm{R}_g = \bm{R}_g\bm{S}$ or, equivalently, 
\[
  \tilde{\bm{S}}\tilde{\bm{R}}_g = \bm{Q}^*\bm{S}\bm{Q}\bm{Q}^*\bm{R}_g\bm{Q} = \bm{Q}^*\bm{R}_g\bm{Q}\bm{Q}^*\bm{S}\bm{Q} = \tilde{\bm{R}}_g\tilde{\bm{S}}
\] 
to show that $\bm{Q}$ performs the irreducible decomposition of $\mathcal{S}$. Writing out $\tilde{\bm{R}}_g$ and $\tilde{\bm{S}}$ more explicitly,
\[
  \tilde{\bm{R}}_g =
    \begin{pmatrix}
     \tilde{\bm{R}}_g^{(1)} \otimes \bm{I}_{n_1} &  &  & \cdots\;\bm{0}  \\
      & \tilde{\bm{R}}_g^{(2)} \otimes \bm{I}_{n_2} &  & \quad\;\;\vdots \\
      \vdots\quad\;\; &   & \ddots &   \\
      \bm{0}\;\cdots &  &  & \tilde{\bm{R}}_g^{(\ell)} \otimes \bm{I}_{n_{\ell}}
    \end{pmatrix}, \,
  \tilde{\bm{S}} = 
    \begin{pmatrix}
    \tilde{\bm{S}}_{11} & \tilde{\bm{S}}_{12} & \cdots & \tilde{\bm{S}}_{1\ell} \\
    \tilde{\bm{S}}_{21} & \tilde{\bm{S}}_{22} & \cdots & \tilde{\bm{S}}_{2\ell} \\
    \vdots  & \vdots  & \ddots & \vdots  \\
    \tilde{\bm{S}}_{\ell1} & \tilde{\bm{S}}_{\ell2} & \cdots & \tilde{\bm{S}}_{\ell\ell} 
    \end{pmatrix},
\]
one can see that the commutativity relation implies $(\tilde{\bm{R}}_g^{(i)}\otimes\bm{I}_{n_i})\tilde{\bm{S}}_{ij} = \tilde{\bm{S}}_{ij} (\tilde{\bm{R}}_g^{(j)}\otimes\bm{I}_{n_j})$. 
According to Schur's Lemma \cite{tinkham2003group,lam2013first}, $\tilde{\bm{S}}_{ij} = \bm{0}$ when $i \neq j$ (i.e., when the irreducible representations $\{\tilde{\bm{R}}_g^{(i)}\}$ and $\{\tilde{\bm{R}}_g^{(j)}\}$ are nonisomorphic), and $\tilde{\bm{S}}_{ij} = \bm{I}_{m_j} \otimes \tilde{\bm{S}}^{(j)}$ when $i=j$, where $\tilde{\bm{S}}^{(j)}$ is an $n_j \times n_j$ complex matrix.
Taken together, 
\begin{equation}
  \tilde{\bm{S}} = \bm{Q}^*\bm{S}\bm{Q} = \bigoplus_{j=1}^{\ell} \left( \bm{I}_{m_j} \otimes \tilde{\bm{S}}^{(j)} \right).
\end{equation}
Thus, $\bm{Q}$ simultaneously block diagonalizes all matrices in $\mathcal{S}$ into blocks of sizes $n_j$, each of multiplicity $m_j$.
Based on the dual relation (b) between the commutant algebras, we see that this is the irreducible decomposition of $\mathcal{S}$.


Accordingly, the IRR transformation can be interpreted within the framework of matrix $*$-algebra: it performs the irreducible decomposition of the matrix $*$-algebra $\mathcal{S}$ formed by all $n \times n$ matrices satisfying the symmetry condition \cref{eq:define_s}.
This interpretation explains the parallel between \cref{eq:A-W,eq:irr}.
However, $\mathcal{S}$ may not always be the best matrix $*$-algebra to work with for the stability analysis of synchronization patterns.

In particular, notice that the matrix $*$-algebra $\mathcal{T}$ generated by the matrix set $\mathcal{A} = \{\bm{E}_1,\cdots,\bm{E}_M,\bm{A}\}$ is always a subalgebra of $\mathcal{S}$, as $\bm{E}_m$ and $\bm{A}$ all share the symmetries defined by the subgroup $G$:
\[
  \bm{R}_g^* \bm{A} \bm{R}_g = \bm{A}, \quad \bm{R}_g^* \bm{E}_m \bm{R}_g = \bm{E}_m \quad \forall \, g \in G \ \text{and} \ 1 \leq m \leq M.
\]
This means that the IRR transformation could be considering the SBD of an unnecessarily large set of matrices, and $\mathcal{S}$ might have a coarser irreducible decomposition than its subalgebra $\mathcal{T}$. 
Thus, an SBD transformation applied directly to $\mathcal{A}$ will always give a block structure on par or finer than the one found by the IRR transformation.

\subsection{Optimal separation of perturbation modes}
\label{sec:perturbation class}

Next, we further characterize the decoupling among perturbation modes achieved by an SBD transformation. 
Given an adjacency matrix $\bm{A}$ and a flow-invariant synchronization pattern described by $\{\bm{E}_m\}$, we divide the perturbation modes into three classes according to their dynamical characteristics: 
\begin{enumerate}[(I)]
\item perturbation modes inside the cluster synchronization subspace;
\item perturbation modes transverse to the cluster synchronization subspace and belonging to a $d$-dimensional (the dimension of a single oscillator) flow-invariant subspace under the variational equation \cref{eq:var-adj-cluster}; 
\item perturbation modes transverse to the cluster synchronization subspace and that do not belong to a $d$-dimensional flow-invariant subspace under the variational equation \cref{eq:var-adj-cluster}.
\end{enumerate}
Class I perturbation modes do not destroy the cluster synchronization pattern, while those from classes II and III do.

From an algebraic point of view, the class I perturbation modes correspond to the $M$-dimensional invariant subspace spanned by the diagonal vectors of matrices in $\{\bm{E}_m\}$.
Perturbation modes of class II are associated with a one-dimensional invariant subspace under the matrix set $\mathcal{A}$, whereas perturbation modes of class III are induced by higher-dimensional invariant subspaces.  

In particular, perturbation modes in class II are localized inside individual clusters and each of them is decoupled from all other perturbation modes.
This is the basis of the so-called {\it isolated desynchronization} \cite{pecora2014cluster}, in which an individual cluster can desynchronize without destroying synchronization in other clusters despite their mutual influence mediated by intercluster coupling.
Class III perturbation modes arise from {\it intertwined clusters} \cite{pecora2014cluster,cho2017stable}. 
Two clusters are intertwined if there exist transverse perturbations inside one cluster that are coupled to transverse perturbations in the other.
It is worth noting that not all transverse perturbations inside intertwined clusters belong to class III, as some of them form $d$-dimensional invariant subspaces on their own and are thus class II. 


An SBD transformation finds the optimal separation of perturbation modes that can be inferred from the network structure and cluster patterns.
In particular, it is guaranteed to separate the parallel perturbations (class I) from the transverse ones (classes II and III), completely decouple the perturbation modes in class II, and separate the ones in class III as much as possible.
This separation is ``robust'' in the sense that it works for any intrinsic dynamics $\bm{F}$ and coupling function $\bm{H}$, since it is induced solely by the algebraic structure of the system.
For some special $\bm{F}$ and $\bm{H}$, the flow-invariant subspaces (under the variational equation) induced by the minimal invariant subspaces (under the matrix set $\mathcal{A}$) may not be minimal and can be further decomposed.\footnote{A flow-invariant subspace is said to be {\it minimal} if it does not contain proper subspaces that are flow invariant.}
But such special flow-invariant subspaces are not robust and will be destroyed by small changes to $\bm{F}$ and/or $\bm{H}$.

\subsection{Treating clusters not induced by symmetry}

The strength of the SBD framework becomes even more evident when the oscillators are diffusively coupled, which is another natural class of coupling schemes \cite{pecora1998master,nicosia2013remote} featured prominently in real systems, such as consensus networks \cite{li2010consensus,olfati2007consensus}. 
This class of systems is a special case of \cref{eq:general} and can be described by
\begin{equation}
  \dot{\bm{x}}_i = \bm{F}(\bm{x}_i) - \sigma \sum_{j=1}^{n} L(i,j) \bm{H}(\bm{x}_j),
  \label{eq:lap-cluster}
\end{equation}
where the Laplacian matrix $\bm{L} = \{L(i,j)\}$ is defined as $L(i,j) = \delta_{ij}\mu_i - A(i,j)$, for $\delta_{ij}$ denoting the Kronecker delta and $\mu_i = \sum_j A(i,j)$ representing the indegree of node $i$. 
When compared to systems with adjacency-matrix coupling, the main difference in those with Laplacian-matrix coupling is that the interaction between two oscillators vanishes when they synchronize. 

As a consequence of the diffusive coupling, additional flow-invariant synchronization patterns can emerge that are not predicted by network symmetry. 
These additional patterns are called {\it Laplacian clusters}, and they can be formed by merging some of the symmetry clusters \cite{sorrentino2016complete}.
Since an adjacency matrix and its corresponding Laplacian matrix have exactly the same symmetry (i.e., $\mathrm{Aut}(\bm{A}) = \mathrm{Aut}(\bm{L})$), the original IRR transformation cannot distinguish the systems described by equations \cref{eq:adj-cluster,eq:lap-cluster}.
Thus, it fails to decouple the parallel and transverse perturbations if applied directly. 
In \cite{sorrentino2016complete}, it was proposed that one can apply the IRR transformation to the adjacency matrix of the diffusive network first, then perform additional local coordinate transformations to account for the merging of symmetry clusters induced by the diffusive coupling. 
This method provides valuable insight, but at the same time it adds an additional layer of complexity on top of the IRR calculations.
In fact, all necessary information for the separation of perturbation modes is already encoded in the Laplacian matrix $\bm{L}$ and Laplacian clusters $\{\mathcal{C}_m\}$.
Accordingly, neither network symmetry nor local coordinate transformations are needed in order to properly decouple the variational equation.

What the IRR transformation misses is the diffusive nature of the Laplacian-matrix coupling. 
Due to the IRR transformation's inability to detect nonsymmetry features (e.g., the zero-row-sum of the Laplacian matrix), one has to perform local coordinate transformations to ``manually'' incorporate that information.
An SBD transformation, on the other hand, does not assume any symmetry a priori. 
It can thus be applied directly to the matrix set $\mathcal{L} = \{\bm{E}_1,\cdots,\bm{E}_M,\bm{L}\}$ and automatically takes the additional features of $\bm{L}$ into account. 
As in the case of adjacency-matrix coupling, an SBD transformation can find the optimal separation of perturbation modes for any flow-invariant synchronization pattern under Laplacian-matrix coupling.

\begin{figure}[t]
\centering
\includegraphics[width=.9\linewidth]{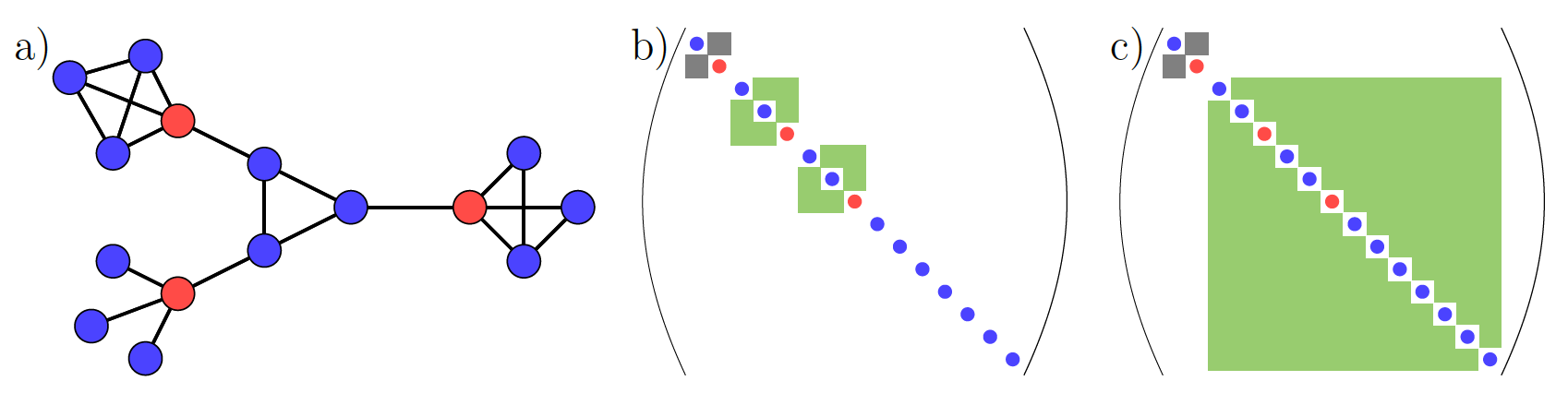}
\vspace{-3mm}
\caption{Input clusters not induced by network symmetry. 
(a) External equitable partition that is not an orbital partition. The partition consists of two clusters---colored red and blue, respectively. The corresponding synchronization pattern is flow invariant under Laplacian-matrix coupling despite there being no symmetry between the blue nodes from different cliques or between the three red nodes.
(b) Common block structure of $\{\bm{E}_1,\bm{E}_2,\bm{L}\}$ found by an SBD transformation. 
Colored circles indicate the cluster each perturbation mode belongs to. 
Gray squares mark the block associated with the parallel perturbations while green squares indicate transverse blocks that are not $1\times1$.
(c) Common block structure produced by the transformation proposed in Ref.~\cite{schaub2016graph}.
}
\label{fig:0}
\end{figure}

More recently, \cite{schaub2016graph} introduced the concept of {\it external equitable partition} as a new way of finding flow-invariant synchronization patterns in Laplacian-matrix coupled systems. 
An external equitable partition splits a network into {\it input clusters} such that each node inside a cluster connects to the same number of nodes in another cluster. 
This definition guarantees that any external equitable partition corresponds to a flow-invariant synchronization pattern and is more general than orbital partitions. 
One example of an external equitable partition that is not an orbital partition is presented in \cref{fig:0}(a).
\cite{schaub2016graph} also proposed the only other widely known symmetry-independent method for the stability analysis of cluster synchronization patterns, which is based on the concept of quotient graphs and uses results from algebraic graph theory. 
While it succeeds in decoupling the parallel and transverse perturbations, in general it fails to further separate the transverse perturbation modes.
In contrast, an SBD transformation not only separates the parallel perturbations from the transverse ones, but it also optimally decouples the transverse perturbations (compare \cref{fig:0}(b) to \cref{fig:0}(c)).
It is still possible to modify the symmetry-based IRR framework for the stability analysis of input clusters by introducing additional local coordinate transformations \cite{siddique2018symmetry}. 
However, the SBD framework can be applied more directly to the problem and, as we show below, leads to a much more scalable algorithm.

\subsection{Computational efficiency and error tolerance}

An SBD transformation is not only directly applicable to more synchronization patterns but is also demonstrably more efficient to compute.  
The computational complexity of the SBD algorithm (\cref{alg:1}) scales with the network size as $\mathcal{O}(n^3)$.
Moreover, the computational cost is independent of the number of symmetries in the network, as demonstrated in \cref{fig:1}(a). 
This gives the SBD algorithm a huge computational advantage over the algorithm based on the IRR transformation \cite{IRRalg}, which relies on the computation of irreducible group representations and becomes inefficient when a large number of symmetries is present.
Since the complete graph with $n$ nodes has the {\it symmetric group} $S_n$ as its automorphism group and $\vert S_n \vert = n!$, the number of symmetries can grow as the factorial of network size $n$.
Combined with the observation that the CPU time scales with the number of symmetries as a power law for the IRR algorithm (orange dots in \cref{fig:1}(a)), it follows that the computational cost of the IRR algorithm can grow superexponentially with the network size. 
This is further illustrated in \cref{fig:1}(b), where the gap between the worst-case CPU time for the two algorithms grows rapidly with $n$, and the IRR algorithm can be more than six orders of magnitude slower than the SBD algorithm even for networks of moderate sizes (e.g., $n=14$).
We note that another polynomial-time algorithm exists, which applies to symmetry clusters \cite{cho2017stable}. 
However, that algorithm was designed to separate clusters that synchronize independently of one another, and thus is not intended to have the same decoupling power as the IRR and SBD algorithms for intertwined clusters and their generalizations.

\begin{figure}[t]
\centering
\includegraphics[width=.55\textwidth]{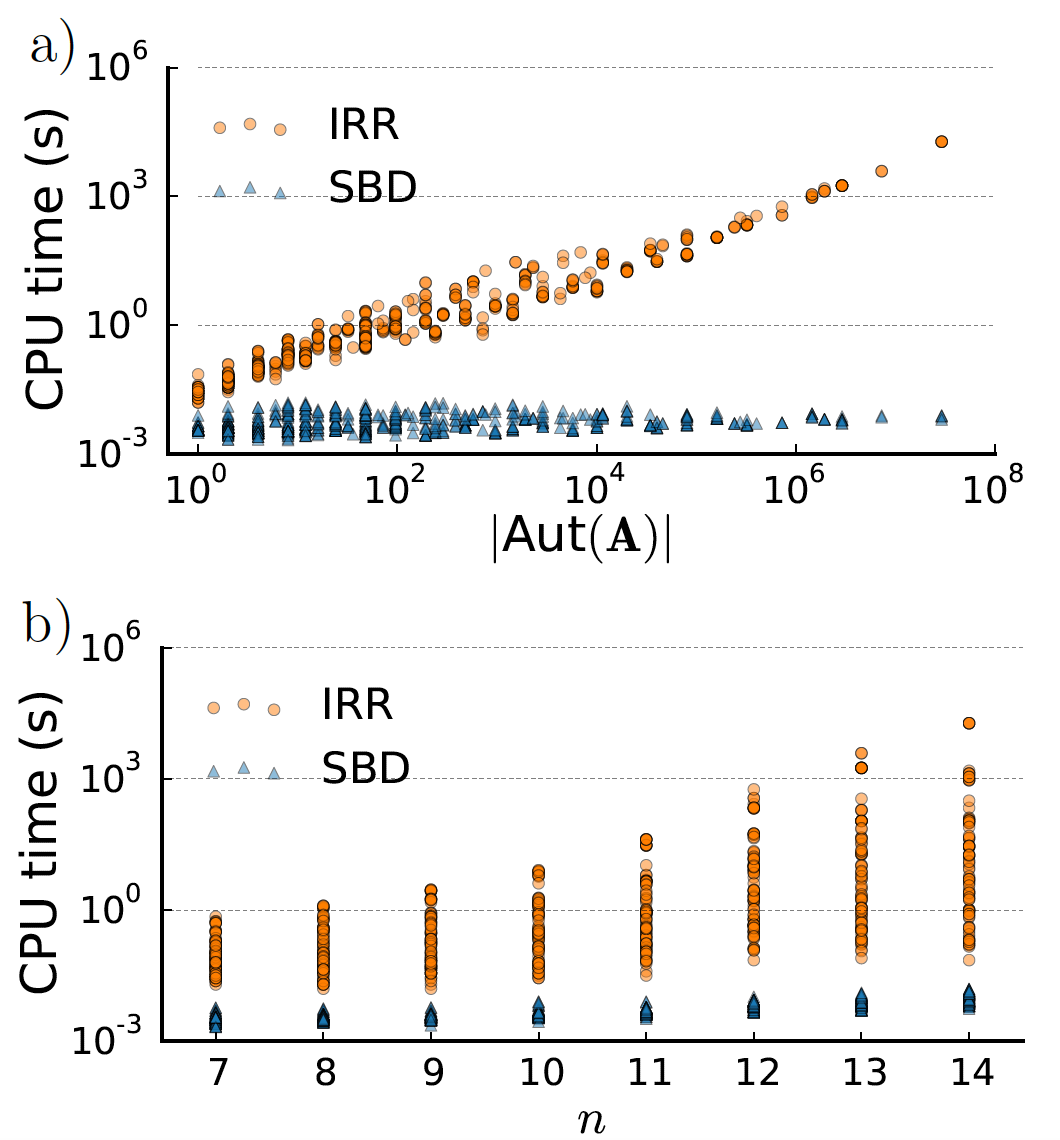}
\vspace{-2mm}
\caption{Comparing the efficiency of the SBD algorithm (\cref{alg:1}) and the IRR algorithm \cite{IRRalg}. 
The tests are done using networks of varying sizes formed by randomly removing 2--10 edges from complete graphs. 
(a) CPU time required to find the IRR transformations (orange) and the SBD transformations (blue) for the symmetry clusters produced by the orbital partition of $\text{Aut}(\bm{A})$, plotted against the number of network symmetries. 
(b) Same data with CPU time plotted as functions of the network size.
All tests were run on an Intel Xeon E5-2680 v3 processor.}
\label{fig:1}
\end{figure}

Another aspect in which the SBD algorithm excels is its error tolerance. 
\Cref{alg:1} can be easily adapted to treat cases in which the coupling matrices contain small errors.
In this case one can simply replace the linear dependence tests by approximate linear dependence tests. 
That is, a vector can be regarded as being linearly independent from a set of vectors if it cannot be expressed as a linear combination of the existing vectors within some preset tolerance.
Unlike the IRR algorithm, which works best when the entries in the adjacency matrix are exact, the SBD algorithm, with its error control capability, has the flexibility to deal with noises and uncertainties in real data.

As an example, we consider a 30-node network generated by deleting six randomly selected edges from a complete graph. 
For each entry of the otherwise binary adjacency matrix, we add a mismatch term drawn from a normal distribution with zero mean and a standard deviation of $10^{-3}$. 
These mismatches can model hardware imperfections and measurement errors in real systems. 
We then equip each node with the dynamics of an electro-optic oscillator used for the first experimental demonstration of chimera states \cite{hagerstrom2012experimental}, described by
\begin{equation}
  \theta_i(t+1) = \left[\beta I(\theta_i(t)) + \sigma \sum_{j=1}^n A(i,j) I(\theta_j(t)) + \xi_i(t) + \delta \right] \text{mod} \: 2\pi,
  \label{eq:opto}
\end{equation}
where $\theta_i$ is the phase for the $i$th oscillator, $\beta$ is the strength of the self-feedback coupling, and $\delta = 0.525$ is introduced to suppress the trivial solution at the origin. 
The nonlinear function $I(\theta) = [1 - \cos(\theta)]/2$ models the dynamics of individual oscillators as well as their interaction function. 
To demonstrate the robustness of our approach in the context of \cref{eq:adj-cluster}, here we also introduce noise terms $\xi_i$ to mimic experimental conditions.
The noise terms are Gaussian, have intensity of $10^{-5}$, and are independent for each oscillator.

\begin{figure}[t]
\centering
\includegraphics[width=1\linewidth]{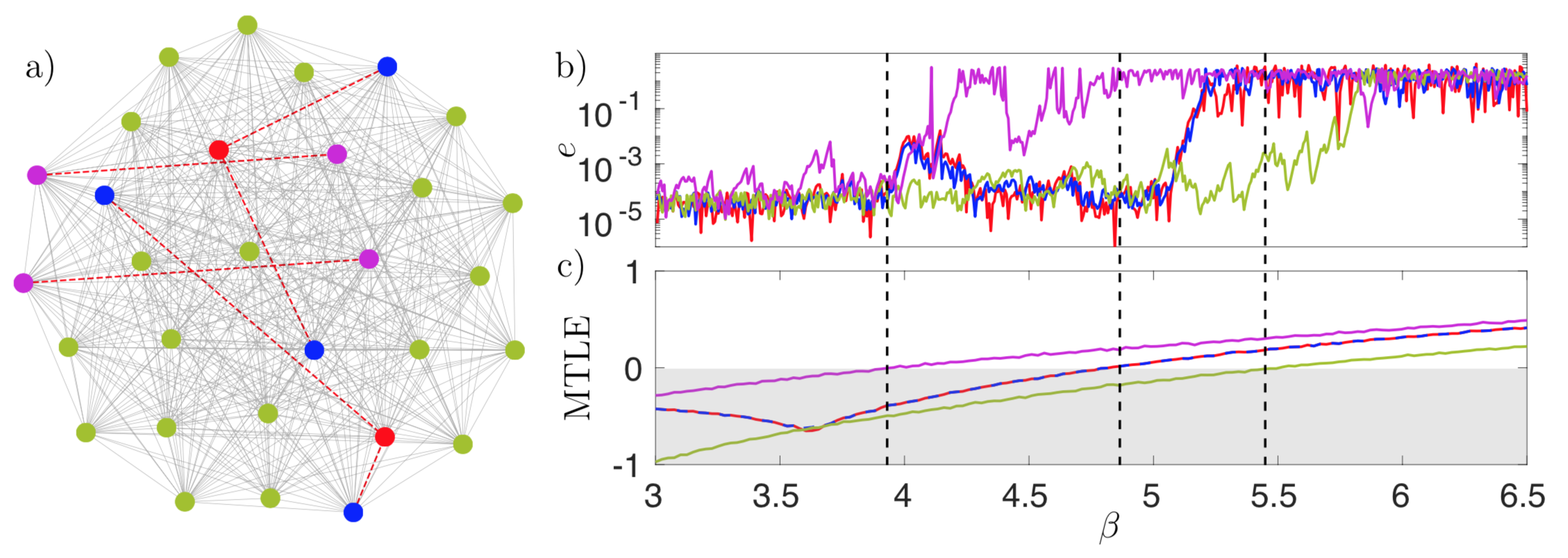}
\vspace{-5mm}
\caption{Symmetry-breaking bifurcations of cluster synchronization patterns in a dense random network.
(a) 30-node network (generated by removing the six red dashed edges from a complete graph) colored according to the orbital partition induced by $\mathrm{Aut}(\bm{A})$. 
(b) Synchronization error for each cluster as the bifurcation parameter $\beta$ is increased slowly from 3 to 6.5.
(c) MTLE for individual clusters calculated from the SBD coordinates.
Both (b) and (c) show a sequence of three desynchronization bifurcations as $\beta$ increases, which are indicated by vertical dashed lines.}
\label{fig:2}
\end{figure}

The network admits a flow-invariant synchronization pattern of four clusters, as shown in \cref{fig:2} (clusters are indicated by node colors), which is induced by the orbital partition of $\mathrm{Aut}(\bm{A})$. 
The IRR algorithm is not practical for this system due to the mismatch terms in the adjacency matrix and the huge number of symmetries present, which is generally the case for dense random networks. 
This particular network has approximately $1.557 \times 10^{20}$ symmetries, and thus extrapolation from \cref{fig:1} suggests around a billion years of CPU time for the IRR algorithm to find the right transformation. 
In contrast, \Cref{alg:1} finds an SBD transformation of $\{\bm{E}_1,\cdots,\bm{E}_4,\bm{A}\}$ within one CPU second. 
Under the SBD coordinates, the matrices share one $4 \times 4$ block corresponding to class I perturbation modes, twenty-four $1 \times 1$ blocks corresponding to class II perturbation modes, and one $2 \times 2$ block corresponding to class III perturbation modes (the red and blue clusters are intertwined).\footnote{See \cref{sec:perturbation class} for the definition of the perturbation classes.}

Based on this decomposition, we calculate the maximal transverse Lyapunov exponent (MTLE) for each cluster over a range of parameter $\beta$. 
We further verify their stabilities by directly simulating equation \cref{eq:opto} for $\beta$ slowly increasing from $3$ to $6.5$ and calculating the synchronization error in each cluster.
We define the synchronization error $e_m$ in the $m$th cluster $\mathcal{C}_m$ with $n_m$ nodes as the standard deviation of the phases $\theta_i$ in that cluster: 
\[
  e_m = \sqrt{\sum_{j\in \mathcal{C}_m} (\theta_j - \bar{\theta})^2/n_m},
\]
where $\bar{\theta} = \sum_{j\in \mathcal{C}_m}\theta_j/n_m$.
\Cref{fig:2}(b) and (c) show a sequence of three symmetry-breaking bifurcations as $\beta$ is increased: it starts with the isolated desynchronization of the magenta cluster around $\beta = 3.9$, followed by the concurrent loss of stability of the red and blue clusters around $\beta = 4.9$, and ends with a transition to incoherence in the green cluster just below $\beta = 5.5$.

\section{Extension to nonidentical oscillators and coupling functions}
\label{sec:unified}

The techniques developed in the previous sections can be easily extended to study cluster synchronization of nonidentical oscillators with disparate coupling functions.
In this section, we establish such a generalized formalism and use it to discover permanently stable chimera states in multilayer networks.

A system of (possibly nonidentical) oscillators diffusively coupled through a multilayer network with $R$ different types of interactions can be described by
\begin{equation}
  \dot{\bm{x}}_i = \bm{F}_{k(i)}(\bm{x}_i) - \sum_{r=1}^{R} \sigma_r \sum_{j=1}^{n} \bm{L}_{r}(i,j) \bm{H}_r(\bm{x}_j),
  \label{eq:7}
\end{equation}
where $\bm{F}_{k(i)}: \mathbb{R}^d \rightarrow \mathbb{R}^d$ 
is the vector field governing the uncoupled dynamics of the $i$th oscillator, $k$ indexes the $K$ different functions $\{\bm{F}_k\}$ that can be assigned to each oscillator, and $\bm{L}_r$ is the Laplacian matrix representing the $r$th type of interaction $\bm{H}_r$.
Other special cases of \cref{eq:general}, corresponding to different choices of the coupling matrices in \cref{eq:7}, can be treated similarly, as outlined below.

For any flow-invariant synchronization pattern, a variational equation governing the evolution of $\delta\bm{X} = (\delta \bm{x}_1^\intercal, \cdots, \delta \bm{x}_n^\intercal)^\intercal$ can be obtained by linearizing \cref{eq:7} around the corresponding cluster synchronization manifold:
\begin{equation}
  \delta\dot{\bm{X}} = \left( \sum_{m,k} \bm{E}_m\bm{D}_k \otimes J\bm{F}_k(\bm{s}_m) - \sum_{m,r} \sigma_r \bm{L}_r\bm{E}_m \otimes J\bm{H}_r(\bm{s}_m) \right) \delta\bm{X},
  \label{eq:var-multi-cluster}
\end{equation}
where $\bm{s}_m$ is the synchronization trajectory of the $m$th cluster. 
Recall that $\bm{E}_m$ is an $n \times n$ diagonal matrix encoding the nodes inside the $m$th cluster.
Similarly, let $\mathcal{N}_k$ be the set of nodes equipped with the $k$th function $\bm{F}_k$. 
Then 
\[
  \bm{D}_k(i,i) =
  \begin{cases}
    1 & \quad \text{if } i \in \mathcal{N}_k, \\
    0 & \quad \text{otherwise} \\
  \end{cases}
\]
are $n \times n$ diagonal matrices encoding the assignment of heterogeneous nodes, whose sum satisfies $\sum_{k=1}^K \bm{D}_k = \bm{I}_n$.
In order to find the coordinates that optimally decouple \cref{eq:var-multi-cluster}, one can apply \cref{alg:1} to the following matrix set: $\{\bm{E}_1,\cdots,\bm{E}_M,\bm{D}_1,\cdots,$ $\bm{D}_K,\bm{L}_1,\cdots,\bm{L}_R\}$.

Our formalism can be used, in particular, to search for permanently stable chimera states in multilayer networks.
Broadly speaking, chimera states and their generalizations refer to states in which coherence and incoherence coexist in a system. 
In the context of coupled oscillators, a network in a chimera state splits into one group of synchronized oscillators and one group of incoherent oscillators \cite{panaggio2015chimera,omel2018mathematics}.
Over the past two decades, chimera states have been shown to be a general phenomenon \cite{kaneko1990clustering,kuramoto2002coexistence,abrams2004chimera,abrams2008solvable,omel2008chimera,yeldesbay2014chimeralike,sethia2014chimera,schmidt2015clustering,ashwin2015weak,semenova2016coherence} that arises robustly in physical systems \cite{hagerstrom2012experimental,tinsley2012chimera,martens2013chimera,bick2017robust,totz2018spiral}.
Meanwhile, multilayer and multiplex networks have recently emerged as suitable descriptions of many real systems \cite{kivela2014multilayer,boccaletti2014structure}. 
In the synchronization community, such networks are often used to represent oscillators coupled through multiple types of interactions \cite{sorrentino2012synchronization,irving2012synchronization,sevilla2016inter,Genioe1601679,blaha2019cluster,tang2019master}.

Given the relevance of these developments, it is of interest to consider chimera and chimera-like states in networks with two or more types of interactions.
There have been previous reports of chimera states in multiplex networks based on numerical simulations \cite{ghosh2016birth,majhi2017chimera}.
However, an analytical treatment of their stability is still lacking.
The formalism developed here bridges this gap, since many chimera and chimera-like states can be seen as special cluster synchronization patterns \cite{hart2016experimental,cho2017stable}.

\begin{figure}[t]
\centering
\includegraphics[width=.7\linewidth]{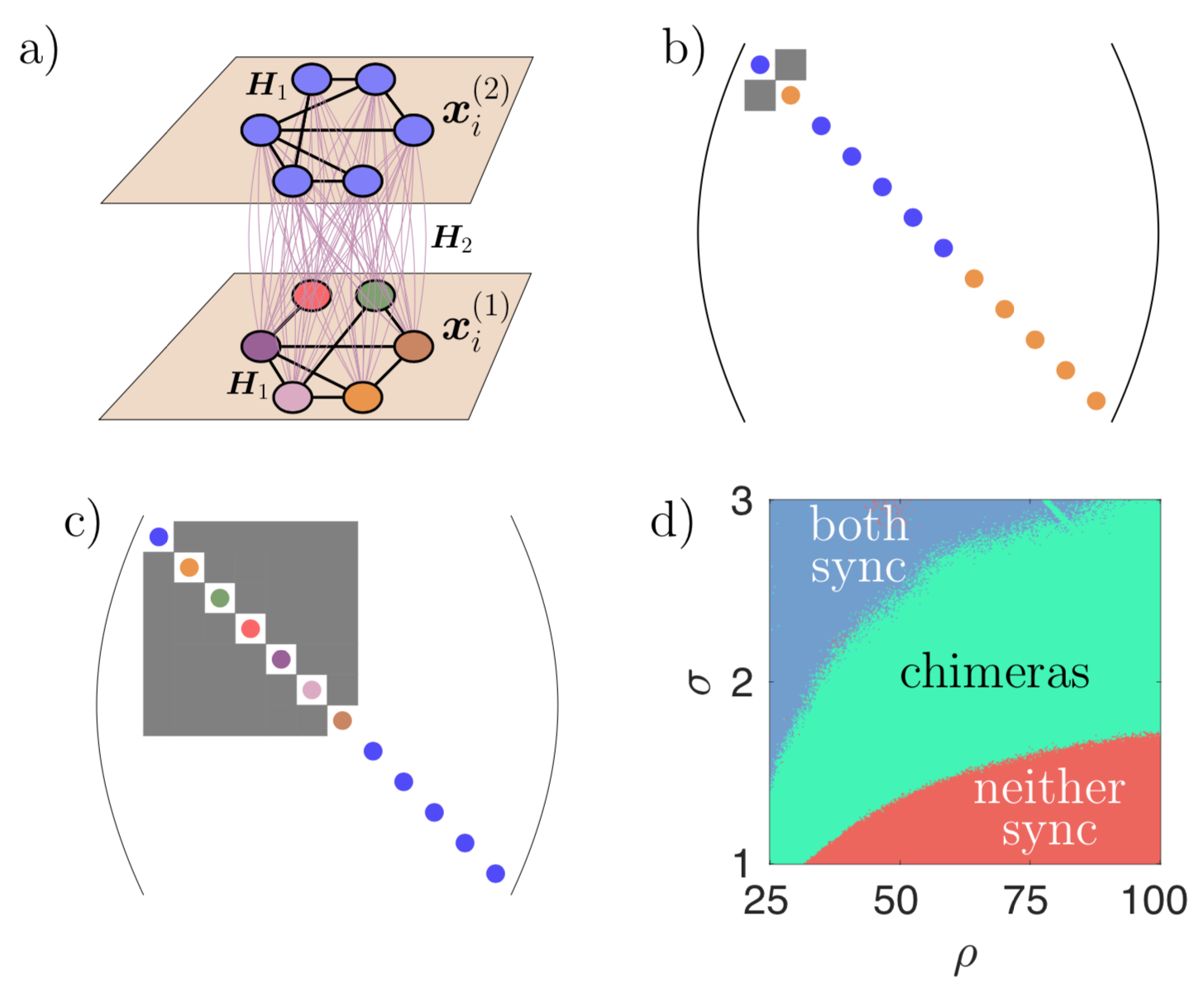}
\vspace{-3mm}
\caption{Chimera states in a multilayer network. (a) Two-layer network of Lorenz oscillators with different intralayer and interlayer interactions, given by $\bm{H}_1 = (0, 0, z)^\intercal$ and $\bm{H}_2 = (0, 0, x)^\intercal$, respectively. The color coded nodes represent a chimera state in which the first layer is incoherent and the second layer is synchronized.
(b) Finest common block structure for the two-cluster state in which both layers are synchronized, which is obtained through an SBD transformation.
(c) Finest common block structure for the seven-cluster state depicted in (a), also obtained through an SBD transformation.
(d) Diagram in the $\rho$-$\sigma$ plane characterizing the stability of the two patterns. 
The three regions correspond to parameters for which both patterns are unstable (red), both patterns are stable (blue), and only the seven-cluster pattern (i.e., chimera state) is stable (green).
}
\label{fig:3}
\end{figure}

As an example, we consider a multilayer network depicted in \cref{fig:3}(a). 
Each layer consists of six identical Lorenz oscillators interacting through eight connections with the coupling function $\bm{H}_1 = (0, 0, z)^\intercal$.
We represent the intralayer connections in the first (second) layer using the Laplacian matrix $\bm{L}_1^{(1)}$ ($\bm{L}_1^{(2)}$).
In addition, the two layers are all-to-all coupled through the coupling function $\bm{H}_2 = (0, 0, x)^\intercal$.
The oscillators in the first layer are thus described by the equations
\begin{equation}
  \begin{split}
    \dot{x}_i^{(1)} = & \;\alpha(y_i^{(1)}-x_i^{(1)}),\\
    \dot{y}_i^{(1)} = & \;x_i^{(1)}(\rho-z_i^{(1)})-y_i^{(1)},\\
    \dot{z}_i^{(1)} = & \;x_i^{(1)}y_i^{(1)}-\beta z_i^{(1)} - \sigma_1\sum_j\bm{L}_1^{(1)}(i,j)z_j^{(1)} + \sigma_2\sum_j(x_j^{(2)}-x_i^{(1)}), 
  \end{split}
\end{equation}
where we set $\alpha=10$, $\beta=2$, $\sigma_1=\sigma$, and $\sigma_2=0.2\sigma$, leaving the parameters $\rho$ and $\sigma$ to be varied.
The oscillators in the second layer are described by similar equations.

To search for chimera states where one layer is synchronized and the other is incoherent, we need to analyze the linear stability of two different cluster synchronization patterns (both formed by input clusters).
Specifically, the two-cluster state in which both layers are coherent ($\bm{x}_1^{(1)}=\cdots=\bm{x}_6^{(1)}$, $\bm{x}_1^{(2)}=\cdots=\bm{x}_6^{(2)}$) should be unstable, while the seven-cluster state $\bm{x}_1^{(1)}\neq\cdots\neq\bm{x}_6^{(1)}, \bm{x}_1^{(2)}=\cdots=\bm{x}_6^{(2)}$ (each cluster represented by a different color in \cref{fig:3}(a)) should be stable.
In both cases, the cluster synchronization manifold can be found by simulating Lorenz oscillators coupled through the corresponding quotient network.
Applying the SBD algorithm to the two-cluster state leads to a common block structure for the matrices in the variational equation \cref{eq:var-multi-cluster}, as shown in \cref{fig:3}(b), where a diagonal entry is colored orange if the corresponding perturbation mode belongs to the first layer and is colored blue if the mode belongs to the second layer.
In this case, the transverse perturbation modes are fully decoupled and the $2\times2$ block corresponds to perturbations inside the cluster synchronization subspace.
Similarly, the common block structure for the seven-cluster state is shown in \cref{fig:3}(c).
In this case, we have a $7\times7$ block representing the parallel perturbations and five $1\times1$ blocks related to the transverse perturbation modes for the coherent layer.
It is straightforward to perform stability analysis under these SBD coordinates. 
We show the results in the $\rho$-$\sigma$ diagram of \cref{fig:3}(d).
Red indicates parameters for which both patterns are unstable, while blue indicates where both patterns are stable.
Chimera states are found in the green region, where only the seven-cluster pattern is stable.

A representative trajectory of the chimera state for $\rho = 60$ and $\sigma = 2$ is shown in \cref{fig:4}.
The lower and upper panels show the dynamics of $x$ variables for oscillators in each layer, while the middle panel shows their respective synchronization errors.
This chimera state is permanently stable and can emerge from random initial conditions.

\begin{figure}[t]
\centering
\includegraphics[width=\linewidth]{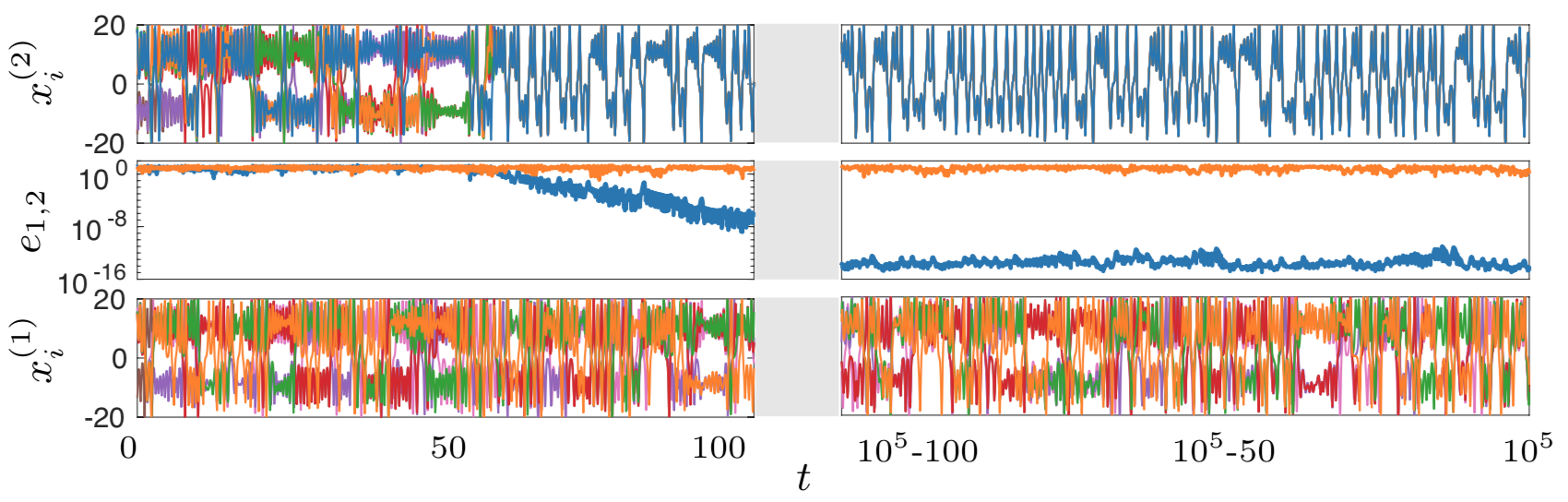}
\vspace{-6mm}
\caption{Trajectories converging to a chimera state in which layer 2 synchronizes and layer 1 remains incoherent for the system in \cref{fig:3}. 
Top and bottom panels: trajectories for oscillators in each layer. 
Each oscillator in one layer is assigned a different color.
When a layer is synchronized, only one color is visible since all trajectories overlap.
Middle panel: synchronization errors $e_1$ for the first layer (orange) and $e_2$ for the second layer (blue).
The parameters used are $\rho = 60$ and $\sigma = 2$.}
\label{fig:4}
\end{figure}

\section{Concluding remarks}
\label{sec:directed}

The framework established here utilizes the finest simultaneous block diagonalization of matrices to study cluster synchronization patterns in complex networks.
This framework has its theoretical foundation rooted in the theory of matrix $*$-algebra and does not rely on symmetry information in the system. 
This results in an algorithm that is faster, simpler, and more robust than the state-of-the-art algorithm based on irreducible representations of network symmetry.
In particular, the SBD framework enjoys the following advantages over the IRR framework and its variants:
\begin{enumerate}
  \item It applies straightforwardly to {\it any} flow-invariant synchronization pattern, including those formed by symmetry clusters, Laplacian clusters, and input clusters.
  \item It can easily treat nonidentical oscillators and oscillators coupled through multiple types of interactions.
  \item It is highly scalable because the SBD transformations can be calculated much more efficiently than the IRR transformations, which enables the stability analysis of complex synchronization patterns in large networks and in networks with a high degree of symmetry.
  \item It is especially suited for practical applications because \cref{alg:1} is robust to uncertainties in the network structure typical of real systems.
\end{enumerate}

A MATLAB implementation of \cref{alg:1} is available online and comes with illustrative examples of use.\footnote{See \url{https://github.com/y-z-zhang/net-sync-sym/}.}
The utility of this algorithm is not limited to network synchronization problems and can be applied, for instance, to reduce the complexity of many problems in which multiple matrices are involved, such as in the control of network systems and in semidefinite programming \cite{murota2010numerical}.

An important open problem for future research concerns the case of directed networks. 
When considering cluster synchronization patterns in directed networks, two complications arise. 
The first concerns the identification of valid clusters. 
Directed networks support many flow-invariant synchronization patterns that do not result from orbital partitions. 
Thus, it is often the case that a synchronization pattern of interest will not be identified by a software based on computational group theory. 
Indeed, any partition of the nodes that satisfies the balanced equivalence relations \cite{golubitsky2006nonlinear,stewart2003symmetry} gives rise to a flow-invariant cluster synchronization pattern. 

The second difficulty involves finding an optimal coordinate system to separate perturbation modes in the stability analysis. 
Since in directed networks the coupling matrices are no longer self-adjoint, one must consider the corresponding matrix algebra (as opposed to the matrix $*$-algebra) to obtain the finest SBD of matrices in the variational equation.
Unlike matrix $*$-algebras, matrix algebras are no longer closed under the conjugate transpose operation and are generally not semisimple algebras. 
This renders the Artin--Wedderburn theorem inapplicable and introduces ``bad'' elements called radicals \cite{lam2013first} such that, in general, matrix algebras cannot be decomposed into the direct sum of irreducible matrix algebras. 
Thus, a promising direction for future research is to generalize the current algorithm to find the finest SBD for matrices that are not necessarily self-adjoint.


\section*{Acknowledgments}
The authors thank Takanori Maehara, Young Sul Cho, Fedor Nazarov, and Jimmy Kim for insightful discussions.

\bibliographystyle{siamplain}
\bibliography{net_dyn}
\end{document}